\documentclass[a4paper,twocolumn,11pt,unpublished]{quantumarticle}
\pdfoutput=1
\usepackage[utf8]{inputenc}
\usepackage[english]{babel}
\usepackage[T1]{fontenc}
\usepackage{amsmath}
\usepackage{amssymb}
\usepackage{amsfonts}
\usepackage{amsthm}
\usepackage{mathrsfs}
\usepackage{color}
\usepackage{multicol}
\usepackage{natbib}
\usepackage{tikz-cd}
\usepackage{hyperref}
\usepackage{tikz}

\newtheorem{theorem}{Theorem}[section]

\newtheorem{remark}{Remark}[theorem]

\newtheorem{definition}{Definition}[theorem]

\newcommand{\catname}[1]{\mathbf{#1}}
\DeclareMathOperator{\Tr}{Tr}
\DeclareMathOperator{\Image}{Im}

\begin{document}

\title{Functoriality of Quantum Resource Theory and Variable-Domain Modal Logic}

\author{Patrick T. Fraser}
\affiliation{Department of Physics, University of Toronto, 60 St. George Street, Toronto, Canada}
\orcid{0000-0003-1804-029X}
\email{p.fraser@mail.utoronto.ca}
\maketitle

\begin{abstract}
Quantum resource theory is a cutting-edge tool used to study practical implementations of quantum mechanical principles under realistic operational constraints. It does this by modelling quantum systems as restricted classes of possible or permissible experimental operations. Modal logic provides a formal tool for studying possibility and impossibility is a completely general logical setting. Here, I show that quantum resource theories may be functorially translated into models of variable-domain S4 modal logic in a way that provides a new class of formal techniques for exploring quantum resource-theoretic problems. I then extend this functorial relationship to an injective one by adding structure to these logical models to reflect the convertibility preorder of resources in the underlying resource theory. It is shown that the resulting logical models are in agreement with the symmetric monoidal category approach to resource theories. I conclude by discussing how this viewpoint may be deployed concretely using purely model-theoretic considerations, dispensing with Hilbert space structure altogether.
\end{abstract}

The operationalization of quantum theory around agent-based informational measures and operational constraints has proven to be one of the most significant conceptual advancements for its practical application. Indeed, operational quantum theory provides the backbone for most applications in quantum cryptography, computing, communication, and many other diverse settings (see, for instance,~\cite{nielsen:2010,wilde:2013,bruss:2019}).

Quantum resource theory essentially takes this lesson -- of formulating quantum theory in terms of constrained classes of physical operations -- and generalizes it to a unified formalism wherein this pragmatic view may be fully realized. This formalism is very powerful and has been successfully applied to in many diverse settings to study entanglement~\cite{vedral:1997}, non-Gaussianity~\cite{genoni:2010} coherence~\cite{baumgratz:2014}, computation~\cite{veitch:2014}, quantum thermodynamics~\cite{ng:2018}, and contextuality~\cite{amaral:2019}, among other things.

In any particular context, a quantum resource theory (QRT) follows a procedure of specifying the kinds of quantum systems under investigation and then specifying a selection of allowed or \textit{possible} operations which can be carried out on these quantum systems. Such a stipulation carries with it a dual notion of disallowed or \textit{impossible} operations (which correspond to resources). These notions of possibility and impossibility -- and modality more generally -- are the essential features of the theory. Importantly, these notions, though formulated in the language of channels between Hilbert spaces, do not strictly require the full Hilbert space structure of the theory; there is an apparent redundancy or over-specification in the theory's representation (for a modern, less redundant approach to quantum theory, see~\cite{coeke:2017,heunen:2020}).

In a completely separate domain of research (mathematical and philosophical logic) questions about the general formalization of possibility, necessity, and modality, have been studied in great detail using modal logic. On the philosophical side, modal logic has been used to provide formal insights into many philosophically interesting questions (some examples may be found in~\cite{Hintikka1962,Lewis:1973,Kripke2007,Sider2010,Burgess2012}). In mathematics literature, modal logic has been realized as being in close connection with intuitionistic logic~\cite{Godel1986} and has been explored using topos theory (e.g. ~\cite{awodey:2014,goldblatt:2006}) topological semantics (e.g.~\cite{kremer:2005,awodey:2008}), homotopy type theory (e.g.~\cite{hottbook:2013,corfield:2020}) and many other such frameworks~\cite{Goldblatt2003} proving it to be a rich source of mathematical structure.

Recently, modal logic has been successfully deployed in the study of a variety of foundational aspects of quantum theory (see, for instance,~\cite{NL2018,boge:2019}). I here seek to make a further contribution to this growing view that modal logic may be exploited to make amenable particular issues in quantum information and quantum foundations. I aim to do so in a manner that makes contact with the categorical approach to quantum mechanics and operational theories (e.g.~\cite{abramsky:2013,coeke:2016b,coeke:2017,heunen:2020}).

Specifically, I here show that the modality present in quantum resource theories (namely, in the scope of \textit{possible} interventions which characterize a quantum resource theory) is in fact enough to recover the majority of their mathematical structure. That is, I show that there is a manner of faithfully functorially interpreting quantum resource theories as models of variable-domain modal logic. By then considering such variable-domain models with an added preorder structure on their global domains (induced by the convertibility preorder on quantum states from quantum resource theory), I show that the given functor may be extended to be injective.

Thus, I establish a precise functorial relation between quantum resource theory and variable-domain modal logic. These results are then placed in the context of the modern discussion of categorical quantum theory -- the symmetric monoidal categorical (SMC) approach to resource theories in particular -- by showing that models in the image of the constructed injective functor admit a canonical construction of an SMC which, in the theoretical setting developed by~\cite{coeke:2016}, has the same features as the original QRT. Specifically, the channel structure is the same, and there is a natural identification of free states and resource states in the original QRT with those in the more general resulting SMC.

After this development, I then provide some preliminary results indicating just how this correspondence may be used to pose resource-theoretic questions in the language of modal logic.

Given the intellectual distance between quantum resource theory and modal logic, I begin with a short introduction to both in Sections~\ref{sec:qrt} and~\ref{sec:modal-logic}, respectively. I then discuss the functoriality of these frameworks and prove several results indicating that the conceptual advance provided here makes sense in Section~\ref{sec:functoriality}. This functorial relation is shown to agree with the interpretation of resource theories as symmetric monoidal categories in Section~\ref{sec:smc-models}. Finally, I discuss the interpretation and concrete application of the logical translation provided here in Section~\ref{sec:discussion} and further explore the possibility of extensions to richer logical settings and different intervention-based physical theories.

\section{Quantum Resource Theories}
\label{sec:qrt}
Quantum resource theory provides a formalism for studying quantum mechanical protocols under different kinds of operational constraints. Briefly, a Quantum Resource Theory (QRT) is a collection of permitted operations on specified quantum systems; states which may be generated under these permitted operations are called free states, and those which may not are called resources.

The typical example of a QRT is where two observers are situated in separate laboratories and they are only able to communicate via classical channels. Then no matter what local quantum mechanical operations they perform, they can never create a quantum state which jointly entangles their separate laboratories. Thus, entangled states are resources. However, if they share such an entangled state to begin with, then still using their local and classical operations, they can carry out teleportation protocols which would otherwise be impossible, hence why such a resource is `resourceful.'

In this basic setting, there is a relevant collection of Hilbert spaces $\textbf{H}=\{\mathbb{C},\mathcal{H}^A,\mathcal{H}^B, \mathcal{H}^A\otimes\mathcal{H}^B\}$ where $\mathcal{H}^A$ and $\mathcal{H}^B$ are the lab Hilbert spaces of the two observers ($\mathbb{C}$ is thr `trivial' space; it's role shall be made clear shortly). There is also a class of permissible operations (i.e. quantum channels) $\mathcal{O}$ which consists of local operations and classical communication (LOCC). Essentially, any channel may be permitted on $\mathcal{H}^A$ or $\mathcal{H}^B$, but only classical communication operations are permitted between the two; thus any operation on $\mathcal{H}^A\otimes\mathcal{H}^B$ must be separable, among other things.

It turns out that the LOCC resource theory is rather sophisticated (see~\cite{chitambar:2014}), however, the above illustration exhibits the basic formula for describing a quantum resource theory; first, one specifies the Hilbert spaces of the systems under consideration (this step is usually left implicit, and often taken to be \textit{all} Hilbert spaces), and then they describe the class of permitted operations on these systems in terms of quantum channels. Once this has been done, there is a class of free states which may be generated using only those permitted operations. Then all other states on the relevant systems are deemed resources (entangled states in the above example).

QRTs operationalize quantum theory by posing questions about quantum systems purely in the pragmatic language of interventions which may be carried out in the laboratory. Indeed, while the above heuristic referred to quantum states, the Choi-Jamio{\l}kowski isomorphism ensures the existence of a full channel-state duality~\cite{jamiolkowsi:1972,choi:1975,jiang:2013}. Thus, it is enough to speak entirely in the language of channels (channels which `prepare' certain states are given as channels from $\mathbb{C}$ to the relevant Hilbert space).

The usual Hilbert space formalism of quantum resource theories~\cite{chitambar:2019} characterizes particular QRTs in terms of \textit{(i)} the available channels by which hypothetical agents are taken to be allowed to intervene on quantum systems, and \textit{(ii)} the collection of states which may be freely prepared and manipulated via those channels. Clearly, \textit{(ii)} is a byproduct of \textit{(i)}; free states are \textit{determined} by free operations. That said, I shall explicitly identify the collection of free states of a given QRT in its definition, even though this is not a primitive defining notion. I shall additionally take the collection of Hilbert spaces on which the relevant quantum systems are defined to be a constitutive feature of QRTs. I proceed with a few definitions.

\begin{definition}
Let $\mathcal{H}$ be a Hilbert space. A \textit{state} $\rho\in B(\mathcal{H})$ is any positive semi-definite, self-adjoint operator with $\Tr\rho=1$.
\end{definition}

The collection of all states on $\mathcal{H}$ is denoted $\mathcal{S}(\mathcal{H})$.

\begin{definition}
A \textit{quantum channel} between systems $A$ and $B$ defined on Hilbert spaces $\mathcal{H}^A$ and $\mathcal{H}^B$ respectively is a completely positive trace preserving (CPTP) map $\Phi:B(\mathcal{H}^A)\to B(\mathcal{H}^B)$.
\end{definition}

A map $\Phi$ is positive if it takes positive operators to positive operators. It is completely positive if for every $k\geq 1$, the induced map $\tilde{\Phi}:M_{k\times k}(B(\mathcal{H^A}))\to M_{k\times k}(B(\mathcal{H^B}))$ which takes $T_{ij}\mapsto \Phi(T_{ij})$ is positive.

\begin{definition}
Given a collection $\mathcal{O}$ of quantum channels over a particular class of Hilbert spaces $\textbf{H}$, the induced \textit{free states} $\mathcal{F}$ are the states $\rho\in \mathcal{S}(\mathcal{H})$ for any $\mathcal{H}\in \textbf{H}$ such that there exists some $\Phi\in\mathcal{O}$ with $\Phi:B(\mathbb{C})\to B(\mathcal{H})$ and $\rho\in \text{Im}(\Phi)$.
\end{definition}

The shorthand notation $\mathcal{O}(A\to B)$ shall be adopted to denote channels of the form $\Phi:B(\mathcal{H}^A)\to B(\mathcal{H}^B)$. Likewise, $\mathcal{O}(A)$ shall denote $\mathcal{O}(A\to A)$ and $1_A$ shall denote the identity channel in $\mathcal{O}(A)$.

We see that $\mathcal{F}$ is determined by $\mathcal{O}$. In the construction of a QRT, one may begin by specifying the states which may be `prepared' (which are obviously free states); such a state $\rho\in\mathcal{S}(\mathcal{H})$ may be assigned a channel $\Phi_\rho\in\mathcal{O}(\mathbb{C}\to\mathcal{H})$ defined by $\Phi_\rho(z)=z\rho$ where $z\in B(\mathbb{C})$ (noting that $B(\mathbb{C})=\mathbb{C}$, whence $\Tr(z)=z$). Since the only choice of $z$ which is a state is $z=1$, we see that $\Phi_\rho$ uniquely chooses $\rho$ out of $\mathcal{S}(\mathcal{H})$ as the only state in its image. It is in this way that $\mathcal{O}(\mathbb{C}\to\mathcal{H})$ naturally defines state preparations and initial free states. Given the rest of free operations of the theory (i.e. the rest of $\mathcal{O}$), the closure of $\mathcal{O}$ under composition shall ensure that any other state which can be obtained by applying free operations to such preparable states will also have a corresponding unique channel in $\mathcal{O}(\mathbb{C}\to\mathcal{H})$. In this way, $\mathcal{F}$ is not so much a free parameter of the theory which may be varied but is a result of the choice of $\mathcal{O}$.

It is known that resource theories may be broadly understood as symmetric monoidal categories~\cite{coeke:2016}. The definition of states as channels out of $\mathbb{C}$ corresponds to the notion that states are processes whose inputs are the trivial object in such a category. When the processes are CPTP maps between complex Hilbert spaces, the trivial object is simply the Hilbert space $\mathbb{C}$ (see, for instance,~\cite[Example 3.2]{coeke:2016}), as illustrated in the previous paragraph.

\begin{definition}
A \textit{quantum resource theory} is a triple $\langle\textbf{H},\mathcal{O},\mathcal{F}\rangle$ of quantum channels and free states over a collection of specified Hilbert spaces $\textbf{H}=\{\mathcal{H}^\alpha\}$ such that, for every Hilbert space $\mathcal{H}^\alpha\in\textbf{H}$, $1_{\alpha}\in\mathcal{O}(\alpha)$ and if $\Phi\in\mathcal{O}(A\to B)$ and $\Psi\in\mathcal{O}(B\to C)$, then $\Psi\circ\Phi\in\mathcal{O}(A\to C)$.
\end{definition}

These conditions on $\mathcal{O}$ ensure that doing nothing is a permitted operation (given by the identity channel), and any two operations may be composed with one another. In a sense, we see already that QRTs must be reflexive (you can always `transform' a state to itself) and transitive (under channel composition). This will serve to be useful in the discussion of modal logic to come, and the reader familiar with the modal system S4 may already see the punchline. (Indeed, the conditions on resource states also reflect the notion of compositionality from category theory, a connection which is not explored here, but which is nevertheless interesting.)

The \textit{resource states} of a QRT are given by the states in $\mathcal{R}=\bigcup_{\alpha}\mathcal{S}(\mathcal{H}^\alpha)-\mathcal{F}$. That is, resources are those things that you cannot freely generate.

A point of clarification is in order. When QRTs are defined, $\textbf{H}$ is not usually considered an explicit feature of the theory. Indeed, in the usual presentation of quantum resource theory~\cite{chitambar:2019}, $\textbf{H}$ is taken to be the class of \textit{all} Hilbert spaces and $\mathcal{O}$ is then taken to contain $1_\mathcal{H}$ for all of these spaces. From that perspective, $\textbf{H}$ as it has been defined here may be viewed as the collection of spaces on which there are \textit{non-trivial} channels (i.e. channels beyond the just the identity, or channels into or out of other spaces). Specifying $\textbf{H}$ explicitly thus doesn't lose any of the structure of the theory, it merely makes it easier to then put in correspondence with modal logic models.

Quantum resource theories under this construal may be isomorphic to one another in the following sense:

\begin{definition}
Two quantum resource theories $\langle\textbf{H},\mathcal{O},\mathcal{F}\rangle$ and $\langle\textbf{H}',\mathcal{O}',\mathcal{F}'\rangle$ are isomorphic if there exists an isomorphism $\varphi_H:\textbf{H}\to\textbf{H}'$ where $\varphi_H(\mathcal{H})\cong\mathcal{H}$ (with $\varphi_{A\to A'}:\mathcal{H}^A\to\varphi_H(\mathcal{H}^A)$ denoting the induced Hilbert space isomorphism\footnote{Note that there is a natural inverse of this isomorphism given by $\varphi_{A'\to A}$.}) such that $\Phi\in\mathcal{O}(A\to B)$ if and only if $\varphi_{B\to B'}\circ\Phi\circ\varphi_{A'\to A}\in\mathcal{O}'$.
\end{definition}

This essentially means that two QRTs $X$ and $Y$ are isomorphic when their collections of Hilbert spaces may be paired up with each other in such a way that every Hilbert space in $X$ is paired with an isomorphic one in $Y$, and the quantum channels in $X$ can be pushed forward to the channels in $Y$. There is no explicit condition relating $\mathcal{F}$ and $\mathcal{F}'$ here, however, none are necessary, as $\mathcal{F}$ is derived entirely from $\mathcal{O}$.

For a pair of such isomorphic QRTs, $\mathcal{O}(A'\to B')$ shall denote $\mathcal{O}(\varphi_H(\mathcal{H}^A)\to\varphi_H(\mathcal{H}^B))$ where $\varphi_H$ is the bijection between their respective classes of Hilbert spaces $\textbf{H}$ and $\textbf{H}'$.

Given QRTs $X=\langle\textbf{H},\mathcal{O},\mathcal{F}\rangle$ and $Y=\langle\textbf{H}',\mathcal{O}',\mathcal{F}'\rangle$, we say that $Y$ is a \textit{sub-QRT} of $X$ and write $Y\leq X$ when $\textbf{H}'\subseteq\textbf{H}$ and $\mathcal{O}'=\mathcal{O}\restriction\textbf{H}'$ (and hence $\mathcal{F}'\subseteq\mathcal{F}\cap(\bigcup_{\mathcal{H}'\in\textbf{H}'}\mathcal{S}(\mathcal{H}'))$, as determined by the restricted class of channels).

I now define a category $\catname{QRT}$ of isomorphism classes of QRTs. To do this, we must determine which maps between QRTs shall constitute arrows in this category. If we let $\catname{CPTP}$ denote the category whose objects are Hilbert spaces and whose arrows are CPTP maps between these Hilbert spaces, then any particular QRT is a subcategory of $\catname{CPTP}$, and there is an inclusion functor $I$ which realizes this fact. Noting that $\textbf{H}$ need not include \textit{all} Hilbert spaces, QRTs need not be \textit{wide} subcategories of $\catname{CPTP}$. While the usual treatment of QRTs in which $\textbf{H}$ \textit{is} all Hilbert spaces, and then QRTs would be wide, the choice to specify $\textbf{H}$ as being only some Hilbert spaces allows for a more natural realization of modal logical models in what follows.

Noting that QRTs are subcategories in this sense, there is an obvious choice of arrows: given two QRTs $X$ and $Y$, we may take the arrows between $X$ and $Y$ to be any morphism $f:X\to Y$ for which the following diagram commutes:

\begin{center}
\begin{tikzcd}
  X \arrow[r, "f"] \arrow[dr, hook, "I"]
    & Y \arrow[d, hook, "I"] \\
    & \catname{CPTP} \end{tikzcd}    
\end{center}

Since $I$ is the obvious inclusion functor, we see that such arrows $f$ are inclusion maps which take sub-QRTs into their super-QRT. That is, if $X$ and $Y$ are QRTs, then we shall say that the arrows $f:X\to Y$ are all of the inclusion maps that map $X$ into an isomorphic sub-QRT of $Y$ (i.e. $X\cong f(X)$ and $f(X)\leq Y$). Thus, inclusions are the natural arrows for the category $\catname{QRT}$.

If each QRT is taken to be a subcategory of $\catname{CPTP}$, then these inclusion morphisms are the usual inclusion \textit{functors} between such subcategories. It is easy to see that the class of inclusions is closed under composition, and includes the identity map $1_X$ which identically maps every element of $X$ onto itself.

There is one final consideration to take note of before $\catname{QRT}$ may be fully defined, which pertains to size. If we were to take $\textbf{H}$ to include \textit{all} Hilbert spaces (as is done in~\cite{chitambar:2019} for instance), then QRTs would be wide subcategories of $\catname{CPTP}$, and thus they would be \textit{large} in the categorical sense (as the collection of all possible Hilbert spaces -- potentially including many copies of the same Hilbert space with different labels -- is too big to be a set). Indeed, many other none-wide QRTs could in principle have $\textbf{H}$ to be a proper class. However, in the translation into modal logic to follow, it shall turn out that the choice of $\textbf{H}$ will correspond to the choice of possible worlds in the resulting logical model. These classes of worlds are generally taken to be sets in usual logic literature, and thus it is convenient to only consider QRTs which are \textit{small} (i.e. \textbf{H} and $\mathcal{O}$ are sets). This detail about possible worlds forming a set shall be discussed further in Section~\ref{sec:modal-logic}.

With these considerations in place, $\catname{QRT}$ shall denote be the category whose objects are isomorphism classes of small QRTs and whose arrows are inclusion maps into sub-QRTs.

In the common quantum resource theory literature,often, one is not so interested in the particular sequence of operations necessary to manipulate resources in a particular way. Rather, they are concerned with understanding \textit{which} resources they may create \textit{provided} they already have some other resource. Ordering resources by their convertibility provides insight into how operationally `valuable' a particular resource is in a given context.

There is a natural preorder~\cite{chitambar:2019} on resource states which reflects exactly this fact. Given a QRT $\langle\textbf{H},\mathcal{O},\mathcal{F}\rangle$ with states $\mathcal{S}$, for any $\rho,\sigma\in\mathcal{S}$, we write $\rho\xrightarrow{\mathcal{O}}\sigma$ if there is some $\Phi\in\mathcal{O}$ with $\sigma=\Phi(\rho)$. Then $\xrightarrow{\mathcal{O}}$ is a preorder on $\mathcal{S}$ which is closed on the class of resource states $\mathcal{R}$ in the sense that if $\sigma$ is a resource state, so too is $\rho$. \cite{coeke:2016} claim that any manner of measuring resources is essentially the same as describing features of this preorder structure. To this end, we shall see in Theorem~\ref{thm:inj-star} that this preorder structure, if added to the modal logic framework described below, is sufficient to completely recover the full structure of any QRT up to isomorphism \textit{without reference to the Hilbert space formalism}.

I now introduce the basic features of variable-domain modal logic.

\section{Modal Logic}
\label{sec:modal-logic}
Modal logic provides a formal setting wherein philosophers and logicians alike can speak formally about possibility and necessity. Essentially, modal logic proceeds by carrying out the usual syntactic constructions for classical logic with an added `possible-worlds' structure and then introduction of additional `modal' operators. These possible worlds represent copies of the underlying classical logic which may be semantically distinct from each other. How these possible worlds are connected then provides a natural interpretation of modal terms using the modal operators. I should note that, while the term `possible worlds' may seem very mystical when first encountered, it merely refers to a particular sort of formal semantics.

Noting that any classical logical connective may be expressed in terms of any other with suitable use of brackets and negation $\neg$, I here suppose for simplicity that the only connective symbol in the logical language to be considered is the conditional $\to$. Likewise, I suppose the usual underlying classical propositional logic axioms and take \textit{modus ponens} to be the only classical rule of inference (there will be an added modal axiom and rule of inference as well).

Possible worlds are constructed in the following manner: one creates a set of `worlds' $\mathscr{W}$ and a binary `accessibility' relation $\mathscr{R}$ on $\mathscr{W}$, and adds two new symbols $\Box$ (the `necessity' operator) and $\Diamond:=\neg\Box\neg$ (the `possibility' operator) to the language. These new operators are interpreted such that $\Box\phi$ means `necessarily $\phi$' and $\Diamond\phi$ means `possibly $\phi$.' These operators are connected to the possible worlds as follows: $\Diamond\phi$ is true at a world $w\in\mathscr{W}$ (that is, $\phi$ is \textit{possible} in $w$) if there is a world $u$ which is accessible to $w$ wherein $\phi$ is true. Necessity of $\phi$ (the formula $\Box\phi$) is interpreted in a similar manner, but instead requiring that $\phi$ is true in \textit{all} worlds which are accessible to $w$.

Truth valuation then occurs at each world separately (though in the modal logic setting considered here, atomic symbols will have a global truth value\footnote{There are many different kinds of modal logic, some of which have world-dependent truth interpretations for atomic symbols, but we are concerned here only with variable-domain modal semantics which uses rigid designators.}). In this article, I am concerned with Variable-Domain Modal Logic (VDML), which may readily be extended to fully quantified predicate modal logic~\cite[pp. 308 -- 314]{Sider2010}. Loosely following~\cite{Sider2010}, a model of this system is defined as follows:

\begin{definition}
A VDML-model is a 5-tuple $\mathcal{M}=\langle\mathscr{W},\mathscr{R},\mathscr{D},\mathscr{Q},\mathscr{I}\rangle$ where:
\begin{itemize}
    \item $\mathscr{W}$ is a non-empty set (called possible worlds).
    \item $\mathscr{R}$ is a binary relation on $\mathscr{W}$ (an inter-word accessibility relation).
    \item $\mathscr{D}$ is a non-empty set (a global domain) of atomic propositional symbols.
    \item $\mathscr{Q}:\mathscr{W}\to\mathcal{P}(\mathscr{D})$ is a function that assigns a sub-domain $\mathscr{D}_w$ to every world $w\in\mathscr{W}$.
    \item $\mathscr{I}:\mathscr{D}\to\{0,1\}$ is a truth interpretation function on atomic symbols.
\end{itemize}
\end{definition}

The domain $\mathscr{D}_w$ of a world $w$ is essentially the restriction of the language of logical discourse available at that world. For formulas that do not contain modal operators, the syntax of this system is given by the usual one for classical proposition logic. When modal operators are present, there is one additional axiom schema for deductions, called the $K$ axiom:

\begin{equation}
    \vdash \Box(\phi\to\psi)\to(\Box\phi\to\Box\psi)
\end{equation}

for all formulas $\phi$ and $\psi$. Likewise, there is one additional rule of inference, called the necessitation rule:

\begin{equation}
    \phi\vdash\Box\phi
\end{equation}

Validity in this system, which is defined at a particular world, is then given by the following valuation function.

\begin{definition}
In a model $\mathcal{M}$, for any atomic symbol $\alpha\in\mathscr{D}$ and formulas $\phi$ and $\psi$ at any world $w\in\mathscr{W}$, the valuation function $V_{\mathcal{M}}$ is given inductively by
\begin{itemize}
    \item $V_{\mathcal{M}}(\alpha,w)=\mathscr{I}(\alpha)$.
    \item $V_{\mathcal{M}}(\neg\phi,w)=1$ iff $V_{\mathcal{M}}(\phi,w)=0$.
    \item $V_{\mathcal{M}}(\phi\to\psi,w)=1$ iff either $V_{\mathcal{M}}(\phi,w)=0$ or $V_{\mathcal{M}}(\psi,w)=1$.
    \item $V_{\mathcal{M}}(\Box\phi,w)=1$ iff for ever $v\in\mathscr{W}$, if $\langle w,v\rangle\in\mathscr{R}$, then $V_{\mathcal{M}}(\phi,v)=1$.
    \item $V_{\mathcal{M}}(\Diamond\phi,w)=1$ iff there exists some $v\in\mathscr{W}$ with $\langle w,v\rangle\in\mathscr{R}$ and $V_{\mathcal{M}}(\phi,v)=1$.
\end{itemize}
\end{definition}

If some formula $\phi$ is valid in a modal $\mathcal{M}$ at \textit{all} worlds of some model $\mathcal{M}$, we write $\vDash_\mathcal{M}\phi$. If $\phi$ is valid in all worlds of all models (for instance, tautologies of classical propositional logic), then we simply write $\vDash\phi$. If $\phi$ is valid in a model $\mathcal{M}$ whenever $\psi$ is valid, we may write $\psi\vDash_\mathcal{M}\phi$.

It should be noted that these models do not include quantification or predication. However, such features may readily be added at a layer of the semantics and syntax which is detached from the underlying structure which is needed for comparison with quantum resource theory. I thus do not define all of this here for the sake of brevity, but note that it is generically possible. VDML-models may be isomorphic in the following way:

\begin{definition}
Two VDML-models $\mathcal{M}=\langle\mathscr{W},\mathscr{R},\mathscr{D},\mathscr{Q},\mathscr{I}\rangle$ and $\mathcal{M}'=\langle\mathscr{W}',\mathscr{R}',\mathscr{D}',\mathscr{Q}',\mathscr{I}'\rangle$ are isomorphic ($\mathcal{M}\cong\mathcal{M}'$) if there exists a pair of bijections $\varphi_W:\mathscr{W}\to\mathscr{W}'$ and $\varphi_D:\mathscr{D}\to\mathscr{D}'$ with

\begin{itemize}
    \item $\langle w,u\rangle\in \mathscr{R}$ iff $\langle \varphi_W(w),\varphi_W(u)\rangle\in\mathscr{R}'$.
    \item $\varphi_D(\mathscr{D}_w)=\mathscr{D}'_{\varphi_W(w)}$
    \item $\mathscr{I}(\alpha)=\mathscr{I}'(\varphi_D(\alpha))$ for all $\alpha\in\mathscr{D}$
\end{itemize}
\end{definition}

This final condition ensures that the truth valuation is equivalent in isomorphic models. In the comparison with quantum resource theories, this will turn out to be necessary for encoding the free states of a QRT in a logical model in such a way that free state structure is preserved under isomorphism.

Given a pair of VDML-models $\mathcal{M}$ and $\mathcal{M}'$, we say that $\mathcal{M}'$ is a \textit{sub-model} of $\mathcal{M}$ and write $\mathcal{M}'\leq\mathcal{M}$ when $\mathscr{W}'\subseteq\mathscr{W}$, $\mathscr{R}'=\mathscr{R}\restriction(\mathscr{W}'\times\mathscr{W}')$, $\mathscr{D}'_w=\mathscr{D}_w$ for all $w\in\mathscr{W}'$, and $\mathscr{I}'(\alpha)=1$ for $\alpha\in\mathscr{D}'$ only if $\mathscr{I}(\alpha)=1$.\footnote{There is an asymmetry in the truth-valuation functions here; this condition is weaker than some might hold want -- some may wish this qualification be `if and only if' instead. However, it is necessary for the discussion to follow, as this will reflect the way in which free states of a QRT change when certain channels are `switched off'. Possible world structures may also be viewed as a particular sort of graph with vertices given by worlds and edges given by the accessibility relation. This asymmetry arises due to the vanishing of some edges in the restriction to a subgraph.} Given two VDML-models $\mathcal{M}$ and $\mathcal{M}'$, an inclusion is a map $g:\mathcal{M}\to\mathcal{M}'$ for which $g(\mathcal{M})\cong\mathcal{M}$ and $g(\mathcal{M})\leq\mathcal{M}'$.

A VDML-model is called a \textit{VDS4-model} if the accessibility relation $\mathscr{R}$ is both reflexive and transitive.\footnote{This alludes to the S4 normal modal system in propositional modal logic. The name `S4' is largely a historical artefact.} The collection of all isomorphism classes of VDS4-models form a category whose arrows are inclusions into larger models which I shall denote $\catname{VDS4}$.

It is worth noting that $\catname{VDS4}$ is a \textit{large} category. To see this, note that each VDS4-model is a Cartesian product of five sets (or functions on sets, which are themselves sets), two of which are arbitrary ($\mathscr{W}$ and $\mathscr{D}$). Thus, the objects of $\catname{VDS4}$ may be enumerated by the collection of all sets, whence $\catname{VDS4}$ and $\catname{Set}$ have the same size. In spite of this, the constitutive features of each VDS4-model are all sets. It is this fact which required we defined quantum resource theories to be only \textit{small} sub-categories of $\catname{CPTP}$ in the previous section.

I now derive a functorial relationship between the categories $\catname{QRT}$ and $\catname{VDS4}$.

\section{Functoriality}
\label{sec:functoriality}
Here, I show that the class of all quantum resource theories is related to the class of all VDS4-models by a functor $F:\catname{QRT}\to\catname{VDS4}$, and determine how close it is to being injective and surjective. I then show that the so-called convertibility preorder of states in a quantum resource theory provides enough additional structure to a VDS4-model to single it out uniquely (up to isomorphism). Thus, I construct an injective functor into the category of VDS4-models with this preorder structure added. I use this to show that the several intuitions about quantum resource theory translate cleanly into the modal logic framework.

Before doing this, I pause to make clear just why this procedure is natural. Essentially, the Hilbert spaces of a QRT correspond to particular quantum systems that agents may intervene upon. Thus, there is a sense in which they may be treated as individual, physically distinct `worlds' whose underlying logical variables are the system states. However, the channels which these quantum agents have access to may allow them to communicate, or otherwise induce interaction between different systems. In this way, there is a certain notion of inter-world accessibility. Thus, the possible-worlds semantics is a natural tool with which to model this phenomenon.

In general, distinct systems have logically distinct states, and so the `logical' language of studying each `world' ought to be distinct. However, given that some systems may be viewed as subsystems, it is natural to suppose that these `logical variables' are not completely separate, but may overlap. It is for this reason that a variable-domain approach is particularly natural. I now make this correspondence precise with a functor I shall denote by $F$.

Let $F$ be defined as follows. Given a QRT $\langle\textbf{H},\mathcal{O},\mathcal{F}\rangle$, let $\mathscr{W}_\textbf{H}=\textbf{H}$ and take $\mathscr{D}_\mathcal{H}=\mathcal{S}(\mathcal{H})$. Then take $\mathscr{R}_\mathcal{O}=\{\langle\mathcal{H}^A,\mathcal{H}^B\rangle|(\exists\Phi\in\mathcal{O}(A\to B))\}$. Take the truth valuation function to then be $\mathscr{I}_\mathcal{F}(\rho)=1$ iff $\rho\in\mathcal{F}$. Take $\mathscr{D}_\textbf{H}=\bigcup_{\mathcal{H}\in\textbf{H}}\mathscr{D}_\mathcal{H}$ and thus $\mathscr{Q}_\textbf{H}(\mathcal{H})=\mathscr{D}_\mathcal{H}$. Then define $F$ by 

\begin{equation}
    F(\langle\textbf{H},\mathcal{O},\mathcal{F}\rangle):=\langle\mathscr{W}_\textbf{H},\mathscr{R}_\mathcal{O},\mathscr{D}_\textbf{H},\mathscr{Q}_\textbf{H},\mathscr{I}_\mathcal{F}\rangle.
\end{equation}

It is easy to check that the image of $F$ is a \textit{VDML}-model since $\catname{QRT}$ contains only small QRTs, and thus $\textbf{H}$ and $\mathcal{O}$ are sets.

\begin{theorem}\label{thm:functoriality}
The map $F:\catname{QRT}\to\catname{VDS4}$ is a functor.
\end{theorem}

\begin{proof}
Since $1_{\alpha}\in\mathcal{O}$ for all $\mathcal{H}^\alpha\in\textbf{H}$, and since $\mathcal{O}$ is transitive, we see that the constructed accessibility relation $\mathscr{R}$ is reflexive and transitive, whence $F(\langle\textbf{H},\mathcal{O},\mathcal{F}\rangle)$ is a VDS4-model (and not just a VDML-model). Thus, functoriality of $F$ amounts to showing \textit{(i)} that the image of $F$ on two isomorphic QRTs are isomorphic as VDS4-models (since objects are defined as isomorphism classes), and \textit{(ii)} that for any arrow $f\in\catname{QRT}$, $F$ induces an arrow $F\circ f\in\catname{VDS4}$.

For \textit{(i)}, let $X$ and $Y$ be QRTs and suppose $X=\langle \textbf{H},\mathcal{O},\mathcal{F}\rangle\cong\langle\textbf{H}',\mathcal{O}',\mathcal{F}'\rangle=Y$. If we denote

\begin{gather*}
    F(X)=\langle\mathscr{W}_\textbf{H},\mathscr{R}_\mathcal{O},\mathscr{D}_\textbf{H},\mathscr{Q}_\textbf{H},\mathscr{I}_\mathcal{F}\rangle\\ F(Y)=\langle\mathscr{W}_{\textbf{H}'},\mathscr{R}_{\mathcal{O}'},\mathscr{D}_{\textbf{H}'},\mathscr{Q}_{\textbf{H}'},\mathscr{I}_{\mathcal{F'}}\rangle
\end{gather*}

we see that QRT isomorphism of $X$ and $Y$ implies the existence of some bijection $\varphi_H$ which defines an isomorphism between $\mathscr{W}_{\textbf{H}}$ and $\mathscr{W}_{\textbf{H}'}$ with $\varphi_H(\mathcal{H})\cong\mathcal{H}$ such that $\Phi\in\mathcal{O}(A\to B)$ if and only if $\varphi_{B\to B' }\circ\Phi\circ\varphi_{A'\to A}\in\mathcal{O}'(A'\to B')$. Immediately, we note that $\varphi_H$ defines an isomorphism between the collections of worlds (Hilbert spaces) of the respective models related by $F$. Moreover, the Hilbert space isomorphisms $\varphi_{A\to A'}$ induced by $\varphi_H$ lift to isomorphisms between $B(\mathcal{H}^A)$ and $B(\varphi_H(\mathcal{H}^A))$, and hence between their sets of states. Then $\varphi_{A\to A'}$ is an isomorphism between domains; we have for every $\mathcal{H}^A\in\textbf{H}$,

\begin{equation}
\begin{split}
    \varphi_{A\to A'}(\mathscr{D}_{\mathcal{H}^A})=&\varphi_{A\to A'}(\mathcal{S}(\mathcal{H}^A))\\
    =&\mathcal{S}(\varphi_{H}(\mathcal{H}^A))\\
    =&\mathscr{D}'_{\varphi_{H}(\mathcal{H}^A)}
\end{split}
\end{equation}

so the respective sub-domains (and hence the respective super-domains) are equivalent in the appropriate sense.

Finally, if $\langle\mathcal{H}^A,\mathcal{H}^B\rangle\in\mathscr{R}_\mathcal{O}$, then there is some $\Phi\in\mathcal{O}(A\to B)$. But then $\varphi_{B\to B' }\circ\Phi\circ\varphi_{A'\to A}\in\mathcal{O}'(A'\to B')$, whence $\langle\varphi_H(\mathcal{H}^A),\varphi_H(\mathcal{H}^B)\rangle\in\mathscr{R}_{\mathcal{O}'}$. But then all three conditions are satisfied ensuring that the resulting VDS4-models are isomorphic. Whence, $F(X)\cong F(Y)$, so $F$ takes isomorphism classes of QRTs to isomorphism classes of VDS4-models.

I now show \textit{(ii)}. Clearly, for any QRT $X$, $F(1_X)=1_{F(X)}$ since 

\begin{equation*}
    F(1_X(X))=F(X)=1_{F(X)}F(X).
\end{equation*}

All remaining arrows are inclusions. Suppose that $X\leq Y$ and thus there is some inclusion arrow $f:X\to Y$ (recalling that objects are defined only up to isomorphism, thus we need not stipulate that $X$ is isomorphic to a sub-QRT). This is the only case when such an arrow exists. Then \textit{(ii)} amounts to showing that $F\circ f$ is an inclusion arrow in $\catname{VDS4}$, and that $F\circ f(X)\leq F(Y)$.

Since $\textbf{H}\subseteq\textbf{H}'$ and $\mathcal{O}=\mathcal{O}'\restriction\textbf{H}$, the induced worlds have $\mathscr{W}=\textbf{H}\subseteq\textbf{H}'=\mathscr{W}'$ and $\mathscr{R}=\mathscr{R}'\restriction(\mathscr{W}\times\mathscr{W})$. Likewise, $\mathscr{D}'_\mathcal{H}=\mathcal{H}=\mathscr{D}_{\mathcal{H}}$ for any $\mathcal{H}\in\textbf{H}$ so the domain condition for VDS4 sub-models is satisfied. Moreover, since $\mathcal{O}\subseteq\mathcal{O}'$, we see that the only change from $\mathcal{F}$ to $\mathcal{F}'$ in the inclusion of $X$ into $Y$ under $f$ is that some additional channels may get `turned on' when $X$ is included, whence $\mathcal{F}\subseteq\mathcal{F}'$ and so $\mathscr{I}_{\mathcal{F}}(\alpha)=1$ only if $\mathscr{I}_{\mathcal{F}'}(\alpha)=1$ for any $\alpha\in\mathscr{D}_{\textbf{H}}$. Thus, $F\circ f(X)\leq F(Y)$ and so $F\circ f$ is an inclusion map in $\catname{VDS4}$.

Thus, $F$ takes objects to object and arrows to arrows, and it trivially obeys compositionality of arrows, since the class of inclusions is closed under composition. Therefore, $F$ is a functor.
\end{proof}

I now determine certain features of this functor. I first show that $F$ is faithful and then determine how close it is to being injective and surjective.

\begin{theorem}\label{thm:faithful}
$F$ is faithful.
\end{theorem}

\begin{proof}
Since the objects of $\catname{QRT}$ and $\catname{VDS4}$ are isomorphism classes, this result is straightforward. To show that $F$ is faithful, let $f$ and $f'$ be any two parallel arrows between the QRTs $X$ and $Y$ (i.e. $f,f'\in\hom_\catname{QRT}(X,Y)$). Then $f(X)\cong X\cong f'(X)$, and so in the isomorphism class of $X$, $f=f'$. That is, there is a \textit{unique} inclusion arrow between any two QRTs up to isomorphism. Thus, it is trivially the case that $F(f)=F(f')$ implies $f=f'$. Hence, $F$ is faithful.
\end{proof}

\begin{theorem}\label{thm:almost-inj}
Let $X=\langle\textbf{H},\mathcal{O},\mathcal{F}\rangle$ and $Y=\langle\textbf{H}',\mathcal{O}',\mathcal{F}'\rangle$ be two QRTs. Then $F(X)\cong F(Y)$ if and only if \textit{(i)} $\textbf{H}\cong\textbf{H}'$, \textit{(ii)} for each $\mathcal{H}^\alpha\in\textbf{H}$, $\mathcal{F}'=\varphi_{\alpha\to \alpha'}(\mathcal{F})$, and \textit{(iii)} for every $\Phi\in\mathcal{O}(A\to B)$, there exists a collection $\{\Psi_\alpha\}\subseteq\mathcal{O}'(A'\to B')$ such that $\Image(\Phi)=\bigcup_\alpha \Image(\varphi_{B'\to B}\circ\Psi_\alpha\circ\varphi_{A\to A'})$.
\end{theorem}

This shows that $F$ is almost injective, but the last condition says that some of the information about the QRT is lost under $F$, whence injectivity fails. I will show in Theorem~\ref{thm:inj-star} that this information is precisely the state convertibility preorder.

\begin{proof}
To begin, if $\textbf{H}\not\cong\textbf{H}'$, then the induced classes of worlds will have $\mathscr{W}_\textbf{H}\not\cong\mathscr{W}_{\textbf{H}'}$, whence $F(X)\not\cong F(Y)$. Likewise, if $\mathcal{F}'\neq\varphi_{\alpha\to \alpha'}(\mathcal{F})$, then the induced truth valuations $\mathscr{I}_\mathcal{F}$ and $\mathscr{I}_{\mathcal{F}'}$ will disagree, whence $F(X)\not\cong F(Y)$. Finally, if condition \textit{(iii)} fails, there will be some channel $\Phi\in\mathcal{O}(A\to B)$ for which $\Psi=\varphi_{B\to B'}\circ\Phi\circ\varphi_{A'\to A}\notin\mathcal{O}'(A'\to B')$, whence $F(X)\not\cong F(Y)$. Therefore conditions \textit{(i) -- (iii)} are necessary.

Conversely, if conditions \textit{(i)--(iii)} are met, then we see that condition \textit{(i)} ensures that $\mathscr{W}_\textbf{H}\cong\mathscr{W}_{\textbf{H}'}$, and likewise, that the respective variable domains are also isomorphic (under the induced maps $\varphi_{\mathcal{H}\to\mathcal{H}'}$), whence the VDML isomorphism conditions on $\mathscr{W}$, $\mathscr{D}$, and $\mathscr{Q}$ are satisfied. Condition \textit{(ii)} then ensures that $\mathscr{I}_F=\mathscr{I}_{F'}\circ\varphi_{\alpha'\to\alpha}$ and so the isomorphism condition on $\mathscr{I}$ will be satisfied. Thus, all that needs to be shown is that $\langle\mathcal{H}^A,\mathcal{H}^B\rangle\in\mathscr{R}_\mathcal{O}$ if and only if $\langle\varphi_H(\mathcal{H}^A),\varphi_H(\mathcal{H}^B)\rangle\in\mathscr{R}_{\mathcal{O}'}$. However, two Hilbert spaces $\mathcal{H}^A$ and $\mathcal{H}^B$ are related by $\mathscr{R}_\mathcal{O}$ if and only if there is some channel $\Phi\in\mathcal{O}(A\to B)$. But condition \textit{(iii)} ensures that, whenever there is such a channel, there is at least one channel (possibly many in the set $\{\Psi_\alpha\}$) in $\mathcal{O}(A'\to B')$. Thus the requisite condition on the relation $\mathscr{R}$ are satisfied as well, whence $F(X)\cong F(Y)$. Therefore conditions \textit{(i) -- (iii)} are sufficient.
\end{proof}

To see the kind of pathology which obstructs injectivity for $F$, consider the following form of counter-example. Let $\textbf{H}=\{\mathbb{C},\mathcal{H}^A,\mathcal{H}^B\}$. Then there is one QRT on this collection of Hilbert spaces, $X$, for which $\mathcal{O}(A\to B)$ only includes one channel $\Phi$ (with possibly other channels out of $\mathbb{C}$). There may, however, be another QRT $Y$ which is identical to $X$ in every way (including the channels from $\mathbb{C}$) except with $\Phi$ replaced with $\xi\circ\Phi$ where $\xi$ is some automorphism on $B(\mathcal{H}^A)$ such that $\mathcal{F}$ is invariant under $\xi$. Then one can readily check $F(X)\cong F(Y)$, even though $A\not\cong B$ (because $\varphi_{B\to B'}\circ\Phi\circ\varphi_{A'\to A}\notin\mathcal{O}(A'\to B')$), whence injectivity of $F$ is violated. However, $X$ and $Y$ \textit{still} satisfy property \textit{(iii)} from the above theorem. These free-state preserving automorphisms seem generically to be the only kind of behaviour which prevents injectivity.

I now determine the obstructions to the surjectivity of $F$. In the VDS4 setting, for a given world $w$, let $T(w):=\{p\in\mathscr{D}_w|\mathscr{I}(p)=1\}$ (i.e. the collection of atomic symbols in the domain of $w$ which are true under the interpretation $\mathscr{I}$). Then we have the following.

\begin{theorem}\label{thm:almost-surj}
Let $\mathcal{M}\in\textbf{VDS4}$. Then there exists some $X\in\textbf{QRT}$ such that $\mathcal{M}=F(X)$ only if \textit{(i)} for all $w,u\in\mathscr{W}$ with $\langle w,u\rangle\in\mathscr{R}$, either $T(u)\neq\emptyset$ or $T(u)=T(w)=\emptyset$, and \textit{(ii)} there is some world $c\in\mathscr{W}$ with $|\mathscr{D}_c|=1$ and $\mathscr{D}_c\cap\mathscr{D}_w=\emptyset$ for $w\neq c$ such that $\langle c,w\rangle\in\mathscr{W}$ for all $w$ with $T(w)\neq\emptyset$.
\end{theorem}

It is here that we see how quantum resource-theoretic considerations limit modality. Condition \textit{(i)} may be understood to mean that truth needs to be monotonically increasing with respect to the accessibility relation $\mathscr{R}$. This reflects the fact that, while resource states may become free, free states can never become resources. Condition \textit{(ii)} essentially reflects that any VDS4-model corresponding to a QRT must have a trivial world which plays the role of the trivial space $\mathbb{C}$.

\begin{proof}
To show that \textit{(i)} is necessary, suppose that $T(u)=\emptyset$ and that $T(w)\neq\emptyset$. That is, $u$ is accessible to $w$, and the domain of $u$ has no atomic symbols which are true under the interpretation function $\mathscr{I}$ of $\mathcal{M}$, while $w$ \textit{does} have true atomic symbols in its domain.

Suppose that $\mathcal{M}=F(X)$ for some QRT $X=\langle \textbf{H},\mathcal{O},\mathcal{F}\rangle$, and let $\mathcal{H}_w$ and $\mathcal{H}_u$ be the Hilbert spaces which are interpreted as worlds $w$ and $u$, respectively, under $F$. The domain $\mathscr{D}_u$ has true symbols under $\mathscr{I}_{\mathcal{F}}$ if and only if it is the case that $\mathcal{S}(\mathcal{H}_u)\cap\mathcal{F}\neq\emptyset$ (i.e. there are free states on $\mathcal{H}_u$). Since $T(u)=\emptyset$, it is thus the case that $\mathcal{F}\cap \mathcal{S}(\mathcal{H}_u)=\emptyset$ and so $\mathcal{S}(\mathcal{H}_u)\subseteq\mathcal{R}$ while $T(w)\neq\emptyset$ implies $\mathcal{F}\cap \mathcal{S}(\mathcal{H}_w)\neq\emptyset$.

Under the assumption that $\langle w,u\rangle\in\mathscr{R}_\mathcal{O}$, there must be some channel $\Phi\in\mathcal{O}(w\to u)$. However, since $\Image\Phi\subseteq\mathcal{S}(\mathcal{H}_u)\subseteq\mathcal{R}$, and $\mathcal{F}\cap\mathcal{S}(\mathcal{H}_w)\neq\emptyset$ there is some free state on $\mathcal{H}_w$ which $\Phi$ takes to a resource state. This is a contradiction. Thus, condition \textit{(i)} is necessary.

To see that \textit{(ii)} is necessary, suppose $X=\langle\textbf{H},\mathcal{O},\mathcal{F}\rangle$ is some QRT. Then if $X$ has any free states, there must be a copy of $\mathbb{C}\in\textbf{H}$ such that, for any free state $\rho\in\mathcal{F}$ on any Hilbert space $\mathcal{H}\in\textbf{H}$, there is a channel $\Phi\in\mathcal{O}(\mathbb{C}\to\mathcal{H})$ with $\rho\in\Image\Phi$. Viewed as a world in $F(X)$, this distinct Hilbert space $\mathbb{C}$ has $|\mathscr{D}_\mathbb{C}|=|\mathcal{S}(\mathbb{C})|=1$ because there is only a single state on $\mathbb{C}$ (the identity operator). Additionally, taken in this way to be a separate space, we see that $\mathscr{D}_\mathbb{C}\cap\mathscr{D}_\mathcal{H}=\mathcal{S}(\mathbb{C})\cap\mathcal{S}(\mathcal{H})=\emptyset$ for all other $\mathcal{H}\in\textbf{H}$. Then any such Hilbert space with free states on it, viewed as a world in the VDS4-model $F(X)$, must be accessible to $\mathbb{C}$ under $\mathscr{R}_\mathcal{O}$. But a Hilbert space $\mathcal{H}\in\textbf{H}$ has free states if and only if $T(\mathcal{H})\neq\emptyset$ in the VDS4-model $F(X)$. Thus, viewed as a world, $\mathbb{C}$ is such a world $c$ in the resulting model $F(X)$ for \textit{any} QRT $X$ with free states. Whence, if a VDS4-model $\mathcal{M}=F(X)$ for some QRT $X$, it must have such a world $c$.
\end{proof}

\begin{remark}
Condition \textit{(ii)} in Theorem~\ref{thm:almost-surj} is not very strict, as such a world $c$ with a single-element domain may be appended to any VDS4-model, and the accessibility of that model may then be extended such to include $\langle c,w\rangle$ for each $w$ with $T(w)\neq\emptyset$.
\end{remark}

For consistency, it is worth stipulating that the unique atomic symbol in the domain of $c$ takes on a truth value of $1$ under $\mathscr{I}$ in models corresponding to the image of a QRT. While not strictly necessary at this time, this will ensure that quantum channels, viewed in the VDS4 setting, are truth preserving. Moreover, it makes sense intuitively to require this; $c$ is essentially a trivial world (just as $\mathbb{C}$ in the QRT setting is a trivial space). Thus, requiring that the symbol in $\mathscr{D}_c$ is true under $\mathscr{I}$ corresponds to supposing that any quantum agent capable of preparing quantum states is first capable of preparing the trivial state from which all other states arise. As such, if we let $p_c$ be the unique element of $\mathscr{D}_c$, we may stipulate that all VDS4-models in the image of $F$ (and thus have such a world $c$ to begin with) satisfy the following axiom:

\begin{equation}\label{eq:c-world-axiom}
    \vdash p_c.
\end{equation}

I now provide some preliminary heuristics for how this sort of modal language may be deployed for studying QRTs, after which I shall elaborate the formalism further to ensure injectivity using the convertibility preorder.

\begin{theorem}\label{thm:free-to-free}
For any QRT $X$ with VDS4-model $F(X)$, if $\rho\in\mathcal{S}(\mathcal{H}^A)$ and $\Phi\in\mathcal{O}(A\to B)$ for some $\mathcal{H}^B$, then $\vDash_{F(X)}\rho\to\Diamond\Phi(\rho)$.
\end{theorem}

\begin{proof}
Whenever $\rho\in\mathcal{F}$, there is some $\Psi\in\mathcal{O}(\mathbb{C}\to A)$ with $\rho\in \Image(\Psi)$. But then if $\Phi:\mathcal{O}(A\to B)$, $\Phi(\rho)\in\mathcal{F}$ as well by transitivity of $\mathcal{O}$. Hence, either $\rho\notin\mathcal{F}$, or both $\rho,\Phi(\rho)\in\mathcal{F}$. The existence of $\Phi$ as a member of $\mathcal{O}$ implies that $\langle\mathcal{H}^A,\mathcal{H}^B\rangle\in\mathscr{R}_{\mathcal{O}}$. From these facts, in the VDS4 model $F(X)$, $V(\rho,\mathcal{H}^A)=1$ only if $V(\Phi(\rho),\mathcal{H}^B)=1$. Then using accessibility, $V(\rho,\mathcal{H}^A)=1$ only if $V(\Diamond\Phi(\rho),\mathcal{H}^A)=1$. Thus $V(\rho\to\Diamond\Phi(\rho),\mathcal{H}^A)=1$, and so $\vDash_{F(X)}\rho\to\Diamond\Phi(\rho)$ for any $\rho$.
\end{proof}

This reflects the modality inherent in the notion that free states are those which are possible to freely generate.

While the functor $F$ clearly allows for a new perspective on quantum resource theory in terms of modal logic, its non-injectivity and non-surjectivity limit its value for clarifying important issues. However, if $\catname{QRT}$ and $\catname{VDS4}$ are appropriately extended with minimal additional structure, $F$ may then be extended to be injective too. Specifically, while the convertibility preorder of a QRT does not appear in its VDS4-model under $F$, if it \textit{is} specified, it makes the functorial correspondence between QRTs and VDS4-models injective. I demonstrate this now.

Suppose a QRT $X=\langle \textbf{H},\mathcal{O},\mathcal{F}\rangle$ is given and that the preorder $\xrightarrow{\mathcal{O}}$ is specified (in principle it may be read off of $X$ directly from $\mathcal{O}$, however, this information is ignored by $F$, so I here mean that we suppose this data is `stored' elsewhere). Then $\xrightarrow{\mathcal{O}}$ is a preorder on $\bigcup_{\alpha}\mathcal{S}(\mathcal{H}^\alpha)=\mathscr{D}_\textbf{H}$ (for $\textbf{H}=\{\mathcal{H}^\alpha\}$), the global domain of $F(X)$ as a VDS4-model.

We may define a pair of new categories, $\catname{QTR}^*$ and $\catname{VDS4}^*$ whose objects are isomorphism classes of QRTs with their convertibility preorders on the states of each Hilbert space in $\textbf{H}$, and VDS4-models together with a preorder on the global domain $\mathscr{D}$, respectively. That is, objects of $\catname{QRT}^*$ are isomorphism classes of pairs $\langle X,\xrightarrow{\mathcal{O}}\rangle$ for QRTs $X$ and objects of $\catname{VDS4}^*$ are isomorphism classes of pairs $\langle\mathcal{M},\preceq\rangle$ for VDS4-models $\mathcal{M}$ with $\preceq$ a preorder on the global domain of $\mathcal{M}$.

Isomorphisms in these categories are just isomorphisms on the unstarred categories with the additional requirement that the preorders are order-isomorphic. Inclusion maps may be defined in the usual manner as before with the additional requirement that the preorder is preserved under inclusion. If we extend $F$ to a new functor $F^*:\catname{QTR}^*\to\catname{VDS4}^*$ which respects this preorder in the right way, then $F^*$ ends up being properly injective.

Given a QRT $X$ with a preorder $\xrightarrow{\mathcal{O}}$, we may define $F^*$ by

$$F^*(\langle X, \xrightarrow{\mathcal{O}}\rangle)=\langle F(X),\preceq\rangle$$

where $\preceq$ is defined as $\rho\preceq\sigma$ (where $\rho$ and $\sigma$ are viewed as symbols in the domain $\mathscr{D}_\textbf{H}$ of $F(X)$) if and only if $\rho\xrightarrow{\mathcal{O}}\sigma$, viewed as states in the QRT $X$. The functoriality of $F$ from Theorem~\ref{thm:functoriality} is enough to show that this is a functor as well.

We have the following theorem.

\begin{theorem}\label{thm:inj-star}
$F^*$ is an injective functor.
\end{theorem}

\begin{proof}
Let $X=\langle\textbf{H},\mathcal{O},\mathcal{F}\rangle$ and $Y=\langle\textbf{H}',\mathcal{O}',\mathcal{F}'\rangle$. Suppose $\langle X,\xrightarrow{\mathcal{O}}\rangle\not\cong \langle Y,\xrightarrow{\mathcal{O}'}\rangle$, it suffices to show that $F^*(\langle X,\xrightarrow{\mathcal{O}}\rangle)\not\cong F^*(\langle Y,\xrightarrow{\mathcal{O}'}\rangle)$. If $X\cong Y$, then $\xrightarrow{\mathcal{O}}$ and $\xrightarrow{\mathcal{O}'}$ are clearly order-isomorphic, whence the images under $F^*$ are isomorphic. If $X\not\cong Y$, then it has previously been established in Theorem~\ref{thm:almost-inj} that either $F(X)\not\cong F(Y)$, or else there is some $\Phi\in\mathcal{O}(A\to B)$ for which there does not a collection $\{\Psi_\alpha\}\subseteq\mathcal{O}'(A'\to B')$ satisfying $\Image(\Phi)=\bigcup_\alpha \Image(\varphi_{B'\to B}\circ\Psi_\alpha\circ\varphi_{A\to A'})$.

If $F(X)\not\cong F(Y)$, then $F^*(\langle X,\xrightarrow{\mathcal{O}}\rangle)\not\cong F^*(\langle Y,\xrightarrow{\mathcal{O}'}\rangle)$, and we are done. Thus, suppose $F(X)\cong F(Y)$. The nonexistence of channels just described then impacts the preorders in the following way. Given the channel $\Phi\in\mathcal{O}$, there is some $\sigma=\Phi(\rho)$ for some $\rho$ such that $\sigma\notin \Image(\varphi_{B'\to B}\circ\Psi\circ\varphi_{A\to A'})$ for any $\Psi\in\mathcal{O}'(A'\to B')$. Thus, $\rho\xrightarrow{\mathcal{O}}\sigma$, but $\varphi_{A\to A'}(\rho)\not\xrightarrow{\mathcal{O}'}\varphi_{B\to B'}(\sigma)$, whence $\xrightarrow{\mathcal{O}}$ and $\xrightarrow{\mathcal{O}'}$ are not order-isomorphic. Thus, the preorders on VDS4-model domains induced under $F^*$ will fail to be order-isomorphic. Whence, $F^*(\langle X,\xrightarrow{\mathcal{O}}\rangle)\not\cong F^*(\langle Y,\xrightarrow{\mathcal{O}'}\rangle)$ as needed. Therefore, $F^*$ is injective.
\end{proof}

One final note which is worth making explicit is that, while it is common to consider only finite-dimensional QRTs, there were no dimensionality assumptions made here; indeed, there wasn't even a cardinality assumption made with respect to the Hilbert space dimension, whence separability is not assumed. This is important because the channel-state duality exploited for the definition of QRTs is generalizable to infinite-dimensional settings (see, for instance,~\cite{holevo:2011}), and there have been approaches to quantum mechanics in more model-theoretically or categorically exotic settings (e.g.~\cite{benci:2019} or~\cite{coeke:2016b}, respectively) so one should strive to provide a description of QRTs which is capable of capturing these generalizations.

QRTs over a finite collection of finite-dimensional systems have corresponding VDS4 models with finite global domains (the $\mathscr{D}_\textbf{H}$'s). If $\textbf{H}=\{\mathcal{H}^\alpha\}$ consists of countably many separable Hilbert spaces, then the global domain will obey $|\mathscr{D}|\leq 2^{\aleph_0}$ with equality when there is at least one Hilbert space of dimension greater than one (whence the collection of states on that space becomes uncountable as it is closed under convex combinations of its pure states). Higher cardinality domains may be needed in non-separable cases.

I now show that the characterization of small QRTs in terms of VDS4-models described in this section is in agreement with the more general theory of resources understood as symmetric monoidal categories.

\section{SMCs for Logical Models}
\label{sec:smc-models}
A more general theory of resources (i.e. beyond just \textit{quantum} resource theories) than that discussed so far is given by~\cite{coeke:2016}. Rather than relying on Hilbert space constructions, they take resource theories to be arbitrary symmetric monoidal categories (SMCs). Following~\cite{awodey:2010}, we have the following definition.

\begin{definition}
A strict symmetric monoidal category is a category $\catname{C}$ with a distinguished unit object 1 that is equipped with a binary operation $\otimes$ which is functorial and associative such that

\begin{gather*}
    x\otimes(y\otimes z)=(x\otimes y)\otimes z\\
    x\otimes 1=x=1\otimes x.
\end{gather*}
\end{definition}

Such a strict monoidal category is \textit{symmetric} if there is also an isomorphism such that $x\otimes y\cong y\otimes x$ for all objects $x$ and $y$. The characterization of resource theories in terms of SMCs given by~\cite{coeke:2016} is more general than the strict case, however, what is to follow here will only require strict SMC, and Mac Lane's strictification theorem~\cite[p. 257]{maclane:1978} ensures that whatever can be proven in the strict case holds also in the non-strict case. 

This section aims to show that the VDS4$^*$-models in the image of $F^*$ described in the previous section are still resource theories in this sense as well. This amounts to showing that there is a canonically constructed SMC associated with each VDS4$^*$-model which corresponds to some QRT under $F^*$ and then to show that the free and resource objects of this SMC agrees with the free and resource states of the underlying QRT.

To each object $\mathcal{M}^*=\langle\mathcal{M},\preceq\rangle$ of the category $\catname{VDS4}^*$, we may define a new category $\catname{M}^*$ whose objects are given by the global domain $\mathscr{D}$ of $\mathcal{M}$. We may then define arrows in $\catname{M}^*$ in the following way. Let the map $f:x\to y$ be defined by $x\mapsto y$ for $x,y\in\mathscr{D}$. Then $f\in\catname{M}^*$ if and only if $x\preceq y$ in $\mathcal{M}^*$ and there exists a pair of worlds $w$ and $u$ such that $x\in\mathscr{D}_w$ and $y\in\mathscr{D}_u$ while $\langle w,u\rangle\in\mathscr{R}$. Since $\mathscr{R}$ is reflexive and transitive for VDS4-models, as is $\preceq$, it follows that there is always an identity map $1:x\to x$ with $1(x)=x$, and that the collection of all such maps is closed under composition. I note that, under this definition, for any two objects $x$ and $y$, $\hom(x,y)$ contains at most a single arrow.

Not all objects in $\catname{VDS4}^*$ correspond to quantum resource theories under $F^*$. Indeed, only those which satisfy conditions \textit{(i)} and \textit{(ii)} of Theorem~\ref{thm:almost-surj} can represent QRTs. Thus, in the discussion to follow, I shall only be concerned with those VDS4-models which satisfy these conditions (the condition on the existence of a distinguished world $c$ in particular). It is these models that shall be shown to have corresponding SMCs which may be interpreted as resource theories. With this in mind, if the VDS4-model $\mathcal{M}$ upon which $\catname{M}^*$ is constructed satisfies the conditions of Theorem~\ref{thm:almost-surj}, I shall denote its corresponding category by $\catname{M}^*_c$.

Given such a category $\catname{M}^*_c$, we may extend it to a new category $\widetilde{\catname{M}}^*_c$ whose objects are arbitrary \textit{finite} conjunctions of the objects in $\catname{M}^*_c$ (which are atomic logical symbols). On $\widetilde{\catname{M}}^*_c$, we may then define a monoidal binary operation $\otimes$ by $x\otimes y\mapsto x\land y$. It is easy to show that that

\begin{equation}
\begin{split}
    x\otimes(y\otimes z)&=x\land(y\land z)=x\land y\land z\\
    =&(x\land y)\land z=(x\otimes y)\otimes z
\end{split}
\end{equation}

where equality is understood to mean syntactic equivalence. Thus the associativity condition is met. Likewise, $x\land y=y\land x$, so symmetry holds (in fact, only isomorphism is necessary, but we have proper equality). The arrows in $\widetilde{\catname{M}}^*_c$ are given by

\begin{equation*}
\begin{split}
    (f_1\otimes\dots\otimes f_n)&(x_1\otimes\dots\otimes x_n)\\
    =&f_1(x_1)\otimes\dots\otimes f_n(x_n)
\end{split}
\end{equation*}

for all $n\geq 1$ and all $x_i$ and $f_i$ in $\catname{M}^*_c$. Composition of arrows in $\widetilde{\catname{M}}^*_c$ is given by

\begin{gather*}
    (f_1\otimes\dots\otimes f_n)\circ(g_1\otimes\dots\otimes g_n)=(f_1\circ g_1\otimes\dots\otimes f_n\circ g_n).
\end{gather*}

All that remains is to equip $\widetilde{\catname{M}}^*_c$ with a unit object. We may take the unique element of $\mathscr{D}_c$ (previously called $p_c$) to be this unit object, which shall be denoted $1$ (where it will be clear from context whether or not $1$ denotes the unity object or an identity map). In treating the object $1$ as a unit, we must now show that $x\otimes 1=x=1\otimes x$ for all $x\in|\widetilde{\catname{M}}^*_c|$. However, this is straightforward, since equality corresponds to syntactic equivalence. Recalling that all models in the image of $F^*$ satisfy the logical axiom~\eqref{eq:c-world-axiom} (which here may be written as $\vdash 1$), we may readily conclude the following facts for any $x\in|\widetilde{\catname{M}}^*_c|$:

\begin{gather*}
x= x,\qquad x\otimes x= x\\
x\otimes 1= x,\qquad x= x\otimes 1.
\end{gather*}

The first two facts show that objects in $|\widetilde{\catname{M}}^*_c|$ behave in a manner analogous to projections in operator theory.\footnote{A different interpretation of this same procedure would be to view $\widetilde{\catname{M}}^*_c$ as a category which looks like a diagonal matrix algebra over symbols in $\mathscr{D}$ (thus taking $\otimes$ to correspond to $\oplus$, the matrix direct sum). In so doing, equality then resembles a sort of `projective equivalence' in the matrix algebra. Thus, the construction given here resembles a logical version of the Grothendiek construction in $K$-theory for $C^*$-algebras, see, for instance,~\cite{rordam:2000}.} The last two facts ensure that $x\otimes 1=x= 1\otimes x$ and so $(\widetilde{\catname{M}}^*_c,\otimes,\circ,1)$ is a strict symmetric monoidal category with respect to the operation $\otimes$. Therefore, as per~\cite{coeke:2016}, $(\widetilde{\catname{M}}^*_c,\otimes,\circ,1)$ is a resource theory where the objects $|\widetilde{\catname{M}}^*_c|$ are states (either resources or free), morphisms are transformations between states, $\circ$ defines the composition of transformations, and $\otimes$ corresponds to the parallel description of states and transformations.

This is not exciting, as the construction was somewhat concocted. What \textit{is} interesting is that if the VDS4$^*$-model $\mathcal{M}^*$ upon which this SMC was constructed corresponds to some QRT $X$ under $F^*$ (i.e. $F^*(X)=\mathcal{M}^*$), then the resulting SMC is equivalent to $X$. Following the way in which morphisms were constructed, we see that the channel data from $X$ is encoded directly into the morphisms of $\widetilde{\catname{M}}^*_c$ via the convertibility preorder. Moreover, following~\cite[Definition 2.2]{coeke:2016}, the free objects of $(\widetilde{\catname{M}}^*_c,\otimes,\circ,1)$ are those objects $x$ for which $\hom(1,x)\neq\emptyset$. But from the definition of 1 (i.e. the world $c$ in Theorem~\ref{thm:almost-surj}), we see that these are exactly those elements of $\mathscr{D}$ for which $\mathscr{I}(x)=1$ (recalling that all such elements lie in worlds which are accessible to $c$). Thus, if $\mathcal{M}^*=F^*(X)$ for some QRT $X$, then the free objects of $(\widetilde{\catname{M}}^*_c,\otimes,\circ,1)$ are exactly the images of the free states of $X$ under $F^*$. In this way, the interpretation of quantum resource theories as logical models also respects the more general characterization of resource theories in terms of SMCs.

\section{Discussion}
\label{sec:discussion}
The construction provided here connects quantum resource theories to variable-domain modal logic in such a way that quantum states are understood as atomic symbols, truth values are understood to refer to the experimental ability to produce those quantum states, and quantum channels are interpreted as modes of accessibility between different quantum systems. I now illustrate why this modal logical representation is valuable for gaining a practical understanding of quantum resource theories.

One immediate justification for the naturalness of this presentation of QRTs is that, provided the VDS4-models carry a pre-ordering on their global domains (i.e. they live in $\catname{VDS4}^*$), injectivity of $F^*$ ensures that \textit{all} of the QRT structure is preserved. Thus, this representation is, in a sense, lossless with respect to its encoding of the qualitative features of the QRTs in question. In this way, there isn't a downside to using the modal logic representation. However, there are several obvious benefits.

First, the modal logic framework doesn't rely on the full Hilbert space and operator-theoretic structure of the usual QRT formalism. Thus, it is mathematically must simpler. One could view the image of $F^*$ in the category $\catname{VDS4}^*$ as a mathematical \textit{reduction} of quantum resource theory, carrying forward only the necessary features of the theory, and employing a language which is more amenable to talking about the sort of modal intervention foundations of the theory.

Another valuable feature of this framework is that the language of modal logic provides a new way to state features of QRTs. As an example of this, let us consider so-called \textit{resource destroying} QRTs (e.g.~\cite{liu:2017}). These are QRTs for which the collection of channels $\mathcal{O}$ may include a channel that takes resource states to free states (the converse, of course, is prohibited by the definition of a resource). In the context of resource \textit{preserving} QRTs, for any such theory $X$ it is the case that if $\rho\not\in\mathcal{F}$, then $\Phi(\rho)\not\in\mathcal{F}$ too. Thus, in the modal logic setting, it is true that $\vDash_{F(X)}\neg\rho\to\neg\Diamond\Phi(\rho)$. Applying contraposition, one obtains $\vDash_{F(X)}\Diamond\Phi(\rho)\to\rho$. This is the converse of theorem~\ref{thm:free-to-free}. Therefore, theorem~\ref{thm:free-to-free} may be extended to a bi-conditional exactly if $X$ is resource-preserving. Resource destroying QRTs, therefore, are those theories for which theorem~\ref{thm:free-to-free} cannot be inverted. Thus, the class of resource-destroying QRTs corresponds exactly to those VDS4-models in the image of $F$ which fail to be models of the extended logical theory which takes the bi-conditional form of~\ref{thm:free-to-free} as an axiom.

There is another useful property which is possessed by the variable domain semantics provided: it may be readily extended to one with full quantification and predication. Indeed, variable domains were constructed specifically to allow modal propositional logic to be extended to a full first-order theory~\cite{garson:2001}. Essentially, to extend a VDML-model to one which is capable of handling quantification and predicates, one needs to add semantical technology for substituting elements of the domain at a world for bound and free variables in quantified expressions. Once this has been done, one may begin carrying out deductions and proving the validity of predicated and quantified formulas. I do not expand the details of this procedure here. However, I shall now discuss how it may be exploited to provide a model-theoretic interpretation of classes of QRTs.

Let us consider \textit{convex} QRTs. These QRTs satisfy the condition that, given any two free states $\rho$ and $\sigma$ on some shared Hilbert space, the convex sum $p\rho+(1-p)\sigma$ is also a free state for all $0\leq p\leq 1$. However, convexity may be cast in terms of (quantified) predicate formulas in the VDS4 setting. Let $\{C_p\}$ be a class of two-place predicates indexed by $p\in[0,1]$ such that $C_p[\rho_1,\rho_2]$ is true exactly when either \textit{(i)} $\rho_1,\rho_2\in\mathcal{F}$ and $p\rho_1+(1-p)\rho_2\in\mathcal{F}$ or \textit{(ii)} one of $\rho_1$ or $\rho_2$ is not in $\mathcal{F}$. Then to show that a particular theory $X$ is convex is equivalent in the modal logic setting to proving 

\begin{equation}\label{eq:convex-pred}
    \vDash_{F(X)}(\forall\rho_1)(\forall\rho_2)(C_p[\rho_1,\rho_2])
\end{equation}

for each value of $p$. If one knows that a particular QRT is convex, then they know that Equation~\eqref{eq:convex-pred} holds for it, whence these formulas may be used in proofs of other QRT results in the modal logical setting.

Another fact to take note of is that the class of all convex QRTs is a sub-collection of the class of VDS4-models which takes $(\forall\rho_1)(\forall\rho_2)(C_p[\rho_1,\rho_2])$ as an axiom schema (for all $p\in[0,1]$). Namely, it is part of the sub-collection which satisfies the conditions of Theorem~\ref{thm:almost-surj}. By then studying the axiomatic extension of VDS4 by these axioms, one may demonstrate features of all such VDS4 models and thus \textit{a fortiori} all convex QRTs. The same sort of analysis may be applied to \textit{affine} QRTs and other classes of QRTs as well.

Yet another facet of quantum resource theory which may be clarified by the modal representation provided here lies in its relation to quantum control. In theoretical idealizations, the introduction of channels to a particular QRT is often trivialized. However, to practically implement a channel, there are usually \textit{necessary} auxiliary control systems that must be employed (for instance, in the quantum thermodynamics setting, see~\cite{woods:2019}). These sorts of necessary control requirements could be used to reduce the class of experimentally feasible QRTs (by discarding all QRTs which are overly-idealized and cannot be practically implemented). If these sorts of control requirements were then stated as VDML axioms, they could be taken as axioms which must be satisfied for a VDS4-model to correspond to a feasible QRT. Then the collection of models of such an axiomatization would yield a new logical setting for investigating feasible operational constraints on quantum systems.

Summarizing the above considerations, the specific details of a \textit{particular} QRT, where much of the richness of the theory lies, may be understood in terms of the validity of certain quantified predicate formulas and formulas involving modal operators in the modal logic setting. Then the properties of these resource theories may be viewed in a model-theoretic way, with no reference to the underlying Hilbert space structure. QRTs which deal with complex aspects of quantum thermodynamics or quantum communication protocols, for instance, may be regarded as specific models in VDS4 which satisfy a particular collection of predicated or modal axioms. In this way, the functorial relation between quantum resource theories and variable-domain modal logic is such that it allows for a new class of tools for the explorations of the operational theory.

Having established a functorial relation between \textit{small} QRTs and VDS4-models, one may naturally be inclined to ask if such a relation holds for \textit{large}, or perhaps \textit{wide} QRTs as well (noting that often QRTs are construed as being wide in $\catname{CPTP}$ anyway). Indeed, it is shown by~\cite{coeke:2016} that wide QRTs, when studied using the SMC formalism, may form rich structures (such as partitioned process theories) which do not obviously translate into small QRTs, and so such a question is not a mere curiosity but may be pertinent to other research questions. I do not answer this question here, but I shall illustrate on the one hand, what such an answer might look like, and on the other, why it is nevertheless often sufficient to consider only small QRTs.

In Section~\ref{sec:modal-logic}, I indicated that \textit{small} QRTs were the only ones whose image under $F$ are VDS4-models as such models are defined in terms of \textit{sets}. While I do not go through this in detail here, I do note that topos theory\footnote{Indeed, topos theory has been successfully used to study aspects of quantum information theory and quantum foundations already, for example in~\cite{isham:1998,isham:2000}.} provides the natural extension to ordinary modal logic in which models may have non-small features (see, for instance,~\cite{goldblatt:2006}). Thus, it is conceivable that wide QRTs may be taken to VDS4-models as constructed within a more general topos semantics. In this setting, the interesting sorts of \textit{wide} (or otherwise large) QRTs one may wish to investigate could then in principle be studied using model theory in this more general context. Topos theory, however, introduces many more complexities which are beyond the scope of this article, though these considerations offer an interesting direction for further research.

For the reader who is inclined to argue that the treatment provided here is too restricted because it only accounts for small QRTs, I remind them that, as described in Section~\ref{sec:qrt}, small QRTs are those in which the collection of Hilbert spaces $\textbf{H}$, and the collection of channels $\mathcal{O}$ are both sets (of arbitrary cardinality, in principle). Any such small QRT may be readily extended to a wide QRT by simply including all other Hilbert spaces into $\textbf{H}$, but only adding their identity channels to $\mathcal{O}$. Thus, small QRTs may be realized as special kinds of wide QRTs (those whose \textit{non-trivial} parts are small) without any loss of generality.

Furthermore, if we consider those QRTs which may \textit{actually} be implemented in the world, we see that they correspond to cases where both $\textbf{H}$ and $\mathcal{O}$ are not only sets, but \textit{finite} sets since any real-world lab may only produce a finite number of quantum systems and may only implement a finite number of operations on those systems. Thus, the idealization to allowing $\textbf{H}$ and $\mathcal{O}$ to be \textit{arbitrary} sets is a significant relaxation already from experimental feasibility. Hence, any QRT which would correspond to a specific operational procedure in a lab (or a strong generalization thereof to arbitrary cardinalities) will be small, and so small QRTs are sufficiently general for essentially any purpose.

I now remark the potential generality of the framework developed here. The programme of theory translation established here need not be strictly limited to quantum resource theories. Indeed, the general principle of studying theories of interventions in the language of modal logic could be applied to many other sorts of theories, at least in principle. Different physical theories may require different semantical considerations (e.g. variable domains may not be enough in general; one may need to introduce additional modal operators or have multi-dimensional semantics, or other such features), and may require different sorts of supplementary structure added to the semantics (in the way quantum resource theory required the addition of a domain preorder corresponding which encoded convertibility). However, the modal logic formalism seems to be sufficiently general to capture the important features of any such intervention-oriented theory such as constructor theory~\cite{deutsch:2013,deutsch:2015} or Oeckl's positive formalism~\cite{oeckl:2019}. Indeed, constructor theory ``seeks to express all fundamental scientific theories in terms of a dichotomy between possible and impossible physical transformations''~\cite[p. 4431]{deutsch:2013}, suggesting a natural modal interpretation.

Suffice to say, the central message of this paper is that the operational turn in physics may be adequately represented in a logical setting completely independent from the details of the underlying formalism of the physical theory being operationalized.

\section{Conclusions}
Here, I presented a basic introduction to both quantum resource theory as an abstract theory of quantum channels, and variable-domain modal logic as a logic of possibility established over a so-called `possible worlds' semantics. Using these constructions, I defined the categories, $\catname{QRT}$ and $\catname{VDS4}$, of isomorphism classes of quantum resource theories and models of variable-domain S4 modal logic, respectively, where the arrows in both cases were inclusions into other quantum resource theories and logical models, respectively.

With this machinery in place, a faithful functor $F:\catname{QRT}\to\catname{VDS4}$ was explicitly constructed and characterized, indicated exactly how close it is to being injective and surjective, and diagnosing precisely which features of quantum resource theory it erases, thereby establishing a strong relationship between the formalisms of these otherwise disparate subjects. By then including state convertibility data in the form of a preorder to construct the categories $\catname{QRT}^*$ and $\catname{VDS4}^*$, it was shown that $F$ could be extended to an injective functor $F^*$. The functorial relation given by $F^*$ was then compared with the symmetric monoidal category approach to resource theory and they were shown to agree.

Finally, it was shown that the modal framework allows certain common intuitions about the general structure of quantum resource theories to be expressed and proven in the modal logic setting. The possibility of applying this functorial relation to the further elaboration of particular quantum resource theories as models of certain axiomatically extended logical theories was discussed.

\section*{Acknowledgements}
The author acknowledge the support of the Natural Sciences and Engineering Research Council of Canada (NSERC), funding reference number USRA-554658-2020. The author is grateful to Martti Karvonen and for helpful comments.

\bibliographystyle{plainnat}
\bibliography{QRT-VDS4-Bib}

\end{document}